\documentclass[journal]{IEEEtran}
\usepackage{cite}
\usepackage{amssymb}
\usepackage{xcolor}
\usepackage{subcaption}
\usepackage{amsthm}  
\newtheorem{theorem}{Theorem}
\newtheorem{lemma}{Lemma} 

\ifCLASSINFOpdf
   \usepackage[pdftex]{graphicx}
\else
   \usepackage[dvips]{graphicx}
\fi
\usepackage{amsmath}
\usepackage{amsfonts} 
\usepackage{balance}

\usepackage{array}
\begin{document}
%
\title{On the Sum of Inverse Gamma RVs and its Application to Wireless Communications Networks}
%
%
%

\author{Petros~S.~Bithas,~\IEEEmembership{Senior Member,~IEEE,}
    Liang~Yang,~\IEEEmembership{Senior Member,~IEEE,} Qiuming~Zhu,~\IEEEmembership{Senior Member,~IEEE,} Hector E. Nistazakis,
and~Daniel~Benevides~Da~Costa,~\IEEEmembership{Senior Member,~IEEE}
\thanks{Copyright (c) 20xx IEEE. Personal use of this material is permitted. However, permission to use this material for any other purposes must be obtained from the IEEE by sending a request to pubs-permissions@ieee.org. }
\thanks{P. S. Bithas is with the Department of Digital Industry Technologies,
National and Kapodistrian University of Athens (NKUA), 15772 Athens, Greece (email: pbithas@uoa.gr).}
\thanks{L. Yang is with the College of Computer Science and Electronic Engineering, Hunan University, Changsha 410082, China (e-mail: liangy@hnu.edu.cn).}
\thanks{Q. Zhu, is with the College of Electronic and Information Engineering, Nanjing University of Aeronautics and Astronautics, Nanjing 211106, China (e-mail: zhuqiuming@nuaa.edu.cn).}
\thanks{H.~E.~Nistazakis is with the Department of Physics, NKUA, Athens, 15784, Greece. (e-mail: { enistaz@phys.uoa.gr}).}
\thanks{D. B. da Costa is with the Interdisciplinary Research Center for
Communication Systems and Sensing (IRC-CSS), Department of Electrical
Engineering, King Fahd University of Petroleum \& Minerals (KFUPM),
Dhahran 31261, Saudi Arabia (e-mail: danielbcosta@ieee.org).}
}
%
%

\markboth{Journal of \LaTeX\ Class Files,~Vol.~14, No.~8, August~2025}%
{Shell \MakeLowercase{\textit{et al.}}: Bare Demo of IEEEtran.cls for IEEE Journals}
%



\maketitle

\begin{abstract}
Investigating the statistics of the sum of random variables (RVs) is fundamental in wireless communication systems, as it can be used to accurately estimate their performance. In this paper, exact closed-form expressions are derived for the probability density function (PDF) and cumulative distribution function (CDF) of the sum of inverse-gamma (IG) RVs. The IG distribution is especially suitable for performance analysis purposes because it provides a tractable yet accurate model of shadowing/local mean power variations. Unlike prior studies on sums of shadowing RVs, the resulting formulas are simple to evaluate and numerically stable. Their utility is demonstrated in two practical settings: unmanned aerial vehicle (UAV)-assisted networks and high-speed railway communications. {Across all considered scenarios, the analytical results are in excellent agreement with the Monte Carlo simulation results,} while the asymptotic approximations closely track the exact results.
\end{abstract}


\enlargethispage{\baselineskip}

%
\IEEEpeerreviewmaketitle

\section{Introduction}
%
%
%
%
\IEEEPARstart{I}{n channel} modeling scenarios, wireless channels exhibit small-scale fading, i.e., rapid amplitude/phase fluctuations from multipath and Doppler. Thus, various continuous distributions have been used for modeling this effect, including Rayleigh, Rician, or Nakagami-$m$. In the same vein, the slow variations induced by large-scale fading are modeled with alternative distribution families, most notably the lognormal \cite{10591508}. {An important distribution that has been recently introduced for shadowing modeling is the  {inverse-gamma (IG)} \cite{8166770}.}  Extensive on-site measurements for wearables, cellular device-to-device, UAV, and vehicle-to-vehicle communication scenarios confirm that IG distribution is an appropriate and accurate shadowing model in emerging wireless applications \cite{bithas2020uav,9380972}. Regardless, based on these models, the sums of random variables (RVs) are widely employed in wireless communications, with their  {probability density function (PDF)/cumulative distribution function (CDF)} statistics enabling analytical studies of detector performance, interference, and diversity combining.

Analytical investigations of the exact statistics of the sum of fading envelopes have been pursued for decades. {Over the past few years, numerous studies have addressed this topic, e.g., \cite{ansari2017new,10003243,chondrogiannis2024harq,garcia2023new}.} 
In \cite{ansari2017new}, the PDF of the sum of square Nakagami-$m$ RVs has been provided in terms of the Fox's $H$ function. Moreover, in \cite{10003243}, exact expressions have been provided in terms of infinite series for the PDF of the sum of $\alpha-\mu$ fading RVs. In  \cite{chondrogiannis2024harq}, the PDF of the sum of generalized-$K$ RVs has been also derived in terms of infinite sums.  {Moreover, in \cite{garcia2023new}, the PDF and the CDF of the sum of independent and identically distributed  (i.i.d.) squared fluctuating two-ray RVs have been analytically derived in terms of infinite series.}  {The analytical investigation of sum statistics becomes more challenging when shadowing is also taken into account. For example, in \cite{badarneh2018sum}, the PDF/CDF of the sum of Fisher--Snedecor \emph{F} variates are expressed in terms of the Lauricella multivariate hypergeometric function, which, although exact, is numerically demanding due to its multivariate nature and increased computational complexity. On the other hand, when focusing exclusively on the shadowing component,} only approximated results are mainly reported, e.g., the PDF of the sum of lognormal RVs has been approximated by another lognormal RV in \cite{hadj2020sinr}. However, despite their simplicity, such approximations do not provide consistent accuracy across all regions of the distribution and parameter settings.

\enlargethispage{\baselineskip}
Although numerically stable PDFs for sums of RVs are essential in wireless analysis, many existing results rely on infinite series or special functions (e.g., Meijer-G/Fox-H), which often increase computational complexity. Consequently, exact and easily computable sum-statistics remain highly desirable. Moreover, despite the growing interest in the IG shadowing model, the distribution of sums of IG RVs remains unexplored. Motivated by the above, the contributions of this paper are:
 {
\begin{itemize}
    \item A novel framework is proposed for sums of independent, not necessarily identically distributed (i.n.d.) IG RVs.
\item Exact, easy-to-evaluate closed-form expressions are derived for the corresponding PDF and CDF.
    \item The impact of IG shaping-parameter quantization on the CDF behavior is investigated.
    \item The computational efficiency of the proposed finite-sum closed-form expressions is assessed.
    \item The usefulness of the framework is demonstrated through two application scenarios.
    \item Firstly, the signal-to-interference ratio (SIR) statistics of UAV-assisted wireless communications are analyzed. 
    \item Secondly, the effect of carrier aggregation (CA) in a high-speed railway (HSR) communication system is studied.
\end{itemize}}

\section{Statistics of the Sum of IG RVs}
Let $I_k$ denote an IG RV with PDF given by \cite{7343336}
\begin{equation}\label{eq:PDF_IG}
    f_{I_k}(x)=\frac{W_k^{\alpha_k }}{\Gamma(\alpha_k) x^{\alpha_k+1}} \exp \left(-\frac{W_k}{x}\right),x> 0,
\end{equation}
where $\alpha_k$ denotes the shaping parameter,  {which is related to the severity of shadowing, with higher values of $\alpha_k$ corresponding to more severe shadowing conditions},  {$W_k$ denotes the scaling parameter of the IG distribution}, and  $\Gamma(\cdot)$ denotes the Gamma function \cite[eq. (8.310)]{B:Ryzhik_book}. The corresponding CDF expression is given by
\begin{equation}\label{eq:IG_marginal_CDF}
    F_{I_k}(x)=\frac{\Gamma \left( \alpha_k,\frac{W_k}{ x}\right)}{\Gamma(\alpha_k) },
\end{equation}
where $\Gamma(\cdot,\cdot)$ denotes the upper incomplete Gamma function \cite[eq. (8.350/2)]{B:Ryzhik_book}. Moreover, the Laplace transform of $I_k$, equivalently the conventional MGF evaluated at $-s$, is defined as $M_{I_k}(s)\triangleq\mathbb{E}\{\exp(-sI_k)\}$, $s\geq 0$. Substituting \eqref{eq:PDF_IG} into this definition, and using \cite[eq. (3.471/9)]{B:Ryzhik_book}, yields
\begin{equation}\label{eq:MGF_Ik}
    M_{I_k} (s)=\frac{2}{\Gamma(\alpha_k)} \left( s W_k\right)^{\frac{\alpha_k}{2}} {\rm K}_{\alpha_k} \left( 2\sqrt{s W_k}\right),
\end{equation}
where ${\rm K}_v(\cdot)$ denotes the modified Bessel function of the second kind and order $v$ \cite[eq. (8.432/1)]{B:Ryzhik_book}. Hereafter, $M_X(s)$ denotes the Laplace transform $\mathbb{E}\{\exp(-sX)\}$. Next, closed-form expressions for the PDF and CDF of the sum of IG distributed RVs are introduced.

\begin{theorem}
Let $\bf{I}$ be a continuous RV defined as
\begin{equation}
    {\bf{I}}=\sum_{k=1}^K I_k,
\end{equation}
where $I_k$ are i.n.d. RVs following the PDF given in \eqref{eq:PDF_IG} and {their shaping parameters are assumed integer
plus one half, i.e., $\alpha_k=n_k+1/2$}  {(with $n_k \in \mathbb{N}$)}\footnote{ {Since the analysis considers i.n.d. RVs, distinct constant values of $\alpha_k$, and thus of $n_k$, may be employed.}}. The PDF of $\bf{I}$ is given by
\begin{equation}\label{eq:PDF_sum}
\begin{split}
    f_{\bf{I}}(x)&=\left( \prod_{k=1}^K \frac{\sqrt{\pi}}{\Gamma(\alpha_k)}\right) \sum_{k_1=0}^{n_1}\cdots\sum_{k_K=0}^{n_K}  \mathcal{A}\\ & \times  \frac{\sqrt{2}/\sqrt{\pi}}{(2x)^{\sum_{k=1}^K \frac{n_k-k_k}{2}+1}}   \exp\left( -\frac{\left( \sum_{k=1}^K {\sqrt{W_k}}\right)^2}{2x}\right) \\ & \times {\rm D}_{1+\sum_{k=1}^K n_k-k_k} \left( \frac{\sum_{k=1}^K \sqrt{2W_k}}{\sqrt{x}}\right),
\end{split}
\end{equation}
where ${\rm D}_v(\cdot)$ denotes the Parabolic cylinder function \cite[eq. (9.240)]{B:Ryzhik_book}  {and $\mathcal{A}=\left[\prod_{k=1}^K\frac{(n_k+k_k)! W_k^{\frac{n_k-k_k}2}}{k_k! (n_k-k_k)!2^{2k_k}} \right]$. It is noted that using \cite[eq. (9.253)]{B:Ryzhik_book}, the parabolic cylinder function can be reduced to the Hermite-polynomial form as follows $D_n(x)=2^{-n/2}\exp\!\left(-\frac{x^2}{4}\right) H_n\!\left(\frac{x}{\sqrt{2}}\right)$.}
\end{theorem}
\begin{proof}
See the Appendix.
\end{proof}\setcounter{lemma}{0} 
\begin{lemma}
The CDF expression of $\bf{I}$ is given by
\begin{equation}\label{eq:CDF_sum}
\begin{split}
    F_{\bf{I}}(x)&=\left( \prod_{k=1}^K \frac{\sqrt{\pi}}{\Gamma(\alpha_k)}\right) \sum_{k_1=0}^{n_1}\cdots\sum_{k_K=0}^{n_K} \mathcal{A}\\ & \times \frac{\sqrt{2}/\sqrt{\pi}}{(2x)^{\sum_{k=1}^K \frac{n_k-k_k}{2}}}   \exp\left( -\frac{\left( \sum_{k=1}^K \sqrt{W_k}\right)^2}{2x}\right) \\ & \times {\rm D}_{-1+\sum_{k=1}^K n_k-k_k} \left( \frac{\sum_{k=1}^K \sqrt{2W_k}}{\sqrt{x}}\right).
\end{split}
\end{equation}
\end{lemma}
\begin{proof}
See the Appendix.
\end{proof}

{To assess the impact of the adopted IG shaping-parameter ''quantization”, we evaluate the absolute CDF difference $|F_{\bf I_1}(\gamma_{\rm th})-F_{\bf I_2}(\gamma_{\rm th})|$, given in \eqref{eq:CDF_sum}, between two configurations with $\alpha_{k_2}=\alpha_{k_1}+1/2$ for the same number of RVs $K$. For $\gamma_{\rm th}=-10~\mathrm{dB}$ (Fig.~\ref{fig:fig1new}), the difference remains very small over most parameter values and becomes noticeable only within a limited region of moderate $K$ and relatively small $\alpha_k$, where a localized increase is observed. For $\gamma_{\rm th}=-5~\mathrm{dB}$ (Fig.~\ref{fig:fig2new}), the discrepancy is more pronounced and reaches its maximum for intermediate values of both $K$ and $\alpha_k$, while it rapidly diminishes outside this region. Overall, the results indicate that the proposed quantization step has a limited effect on the CDF behavior and mainly impacts specific operating regimes, remaining negligible across broad ranges of practical system parameters.}

\begin{figure}[t]
    \centering
    \begin{subfigure}[t]{0.48\columnwidth}
        \centering
        \includegraphics[width=\linewidth]{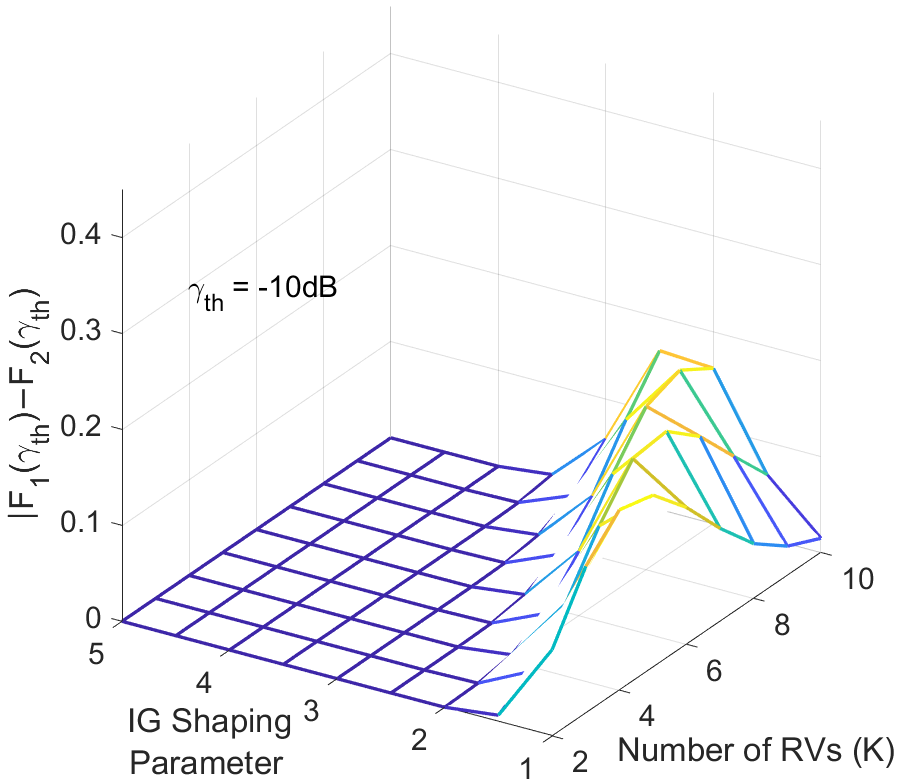}
        \caption{$\gamma_{\rm th}=-10~\mathrm{dB}$.}
        \label{fig:fig1new}
    \end{subfigure}
    \hfill
    \begin{subfigure}[t]{0.48\columnwidth}
        \centering
        \includegraphics[width=\linewidth]{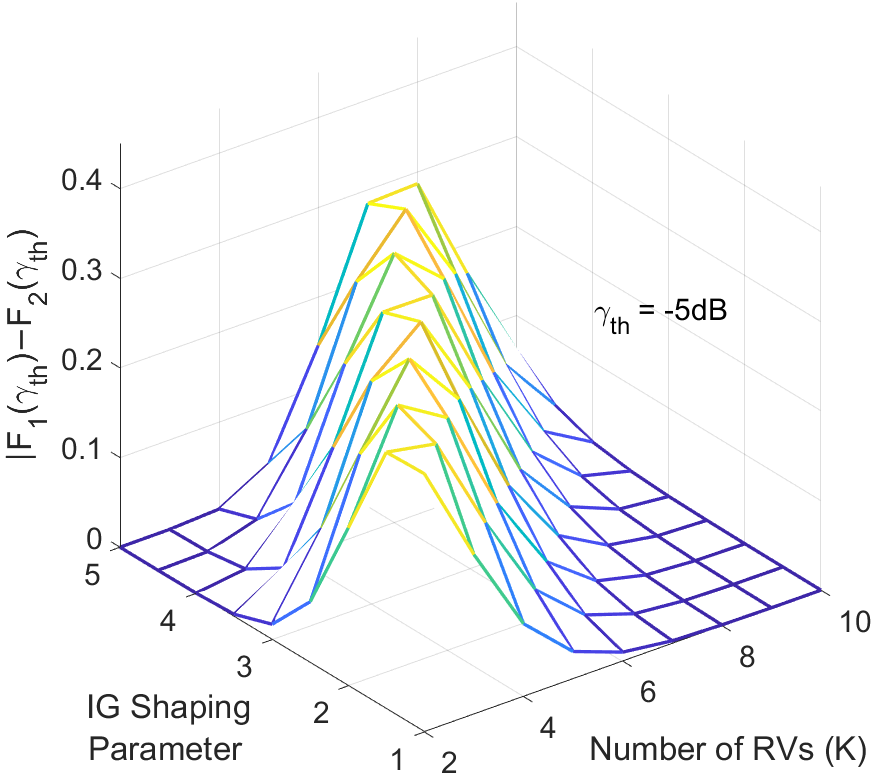}
        \caption{$\gamma_{\rm th}=-5~\mathrm{dB}$.}
        \label{fig:fig2new}
    \end{subfigure}
    \caption{{Absolute CDF difference $|F_{\bf I_1}(\gamma_{\rm th})-F_{\bf I_2}(\gamma_{\rm th})|$ versus the IG shaping parameter and the number of RVs $K$.}}
    \label{fig:twofigs_quantization}
\end{figure}
{Furthermore, we benchmarked the time required to evaluate the PDF of the sum of two IG RVs over a discrete $(x,\alpha_k)$ grid, comparing the proposed closed-form finite-sum expression against a Talbot-based numerical inverse Laplace transform (NILT). Using Wolfram Mathematica, execution times were recorded via repeated evaluations and summarized through the \textit{RepeatedTiming} function. In Table~\ref{tab:timing_ratio_alpha1_alpha3}, assuming $W_k=1$, it is shown that the normalized timing ratio $(t_{\mathrm{NILT}}/t_{\mathrm{closed}})$ was consistently larger than unity, demonstrating that NILT is substantially more time-consuming across all tested configurations, with the gap worsening as the number of summed variables increases.}

{Further investigations have shown that the number of summed RVs $(K)$ affects the normalized timing ratio $(t_{\rm NILT}/t_{\rm closed})$, and hence the relative computational advantage of the proposed closed-form formulation. Overall, the ratio remains predominantly larger than unity, indicating that the closed-form finite-sum expression is typically faster than NILT, although the gap tends to narrow as $(K)$ increases due to higher evaluation complexity. Moreover, while NILT may exhibit nearly constant runtime per point, its robustness can be sensitive to contour settings, numerical precision, and tail regimes in complex scenarios. Beyond timing, the closed-form approach remains attractive because it provides reproducible evaluation and enables further analytical developments (e.g., SIR-based metrics), which are difficult to obtain by combining numerical inversion with additional numerical integrations.
}

\begin{table}[t]
\centering
\caption{{Normalized timing ratio $t_{\rm NILT}/t_{\rm closed}$ for $n_k\in\{1,3\}$ and selected $x$ values.}}
\label{tab:timing_ratio_alpha1_alpha3}
\renewcommand{\arraystretch}{1.1}
\setlength{\tabcolsep}{5pt}
\begin{tabular}{c c c c c c c c}
\hline
$n_k$ & $x=1$ & $x=5$ & $x=9$ & $x=13$ & $x=17$ & $x=21$ & $x=25$ \\
\hline
$1$ & 9.63 & 20.07 & 23.89 & 24.31 & 23.70 & 20.88 & 26.65 \\
$3$ & 15.31 & 12.21 & 12.71 & 12.03 & 9.13 & 10.73 & 12.72 \\
\hline
\end{tabular}
\end{table}

\section{Application Scenarios}
 {Capitalizing on the analytical results presented in Section II, two application scenarios are examined that have recently gained significant interest in the research community}, i.e., a UAV-assisted network and an HSR communication system.  {The performance of both schemes under consideration is evaluated using the criteria of outage probability (OP) and coverage probability (CP).  {The OP is defined as the probability that the instantaneous received SIR $\gamma_{\rm out}$ (or signal-to-noise ratio (SNR) in noise-limited scenarios)  falls below a target threshold $\gamma_{\rm th}$, i.e.,
\begin{equation}
P_{\rm out} \triangleq \Pr\left\{\gamma_{\rm out}\le \gamma_{\rm th}\right\}
=F_{\gamma_{\rm out}}\!\left(\gamma_{\rm th}\right).
\end{equation}
The CP is complementary to the OP, i.e., $P_{\rm cov}=1-P_{\rm out}$.}}

\subsection{UAV-Assisted Wireless Network}
We examine a multi‑tier UAV-assisted communication network, in which UAVs function as airborne base stations (BSs) providing services to ground‑based user equipment (UE) \cite{wang2023outage}. 
A given UE is associated with its \emph{serving} aerial BS, labeled~$m$, from which it receives the desired signal \cite{hadj2020sinr}. 
At the same time, the UE is subjected to co--channel interference originating from $K$ neighboring UAV-BSs, indexed by $k \in \{1,2,\ldots,K\}$ with $k\notin m$. Assuming that the fast fading effects are averaged out by the receiver \cite{borralho2023coverage}, the received SIR can be written as
\begin{equation}\label{eq:Application1_def}
  \gamma_{\rm  out}
  \;=\;
  \frac{P_m G_m d_m^{-n} I_m}
       {\displaystyle\sum_{k=1}^{K} P_k G_k d_k^{-n} I_k}=\frac{I_D}{I_I},
\end{equation}
where $P_i$, with $i \in \{k,m\} $, denotes the transmit power of UAV-BS~$i$, $G_i$ is the antenna gain, $d_i$  {denotes the distance between the $i$-th UAV and the UE, $I_{i}$ represents the shadowing factor of the $i$-th downlink} (with PDF given by \eqref{eq:PDF_IG}), and $n$ denotes the path loss exponent. Based on the definition provided in \eqref{eq:Application1_def} and making a change of variables, the CDF of {desired received signal} $I_D$ is given by $F_{I_D}(\gamma)=F_{I_k}\left(\gamma \right)$, with $W_D=\frac{P_m G_m}{ d_m^{n} }$, while the PDF of {total interfering signals} $I_I$ is given by $f_{I_I}(\gamma)=f_{\mathbf{I}} (\gamma)$, with $W_k=\frac{ P_k G_k}{d_k^{n}}$.  {It is noted that $W_{D}$ denotes the scaling parameter of the IG distribution related to the channel of the serving aerial BS}. Using the laws of probability theory, the CDF of $\gamma_{\rm out}$ can be evaluated as
\begin{equation}\label{eq:USE_general}
    \begin{split}
   &     F_{\gamma_{\rm out}} (\gamma)=\int_0^\infty F_{I_D} \left(\gamma x\right) f_{{I_I}}(x)dx= \\ &
        \frac{\left( \prod_{k=1}^K \frac{\sqrt{\pi}}{\Gamma(\alpha_k)}\right)}{{\Gamma(\alpha_D) }} \sum_{k_1=0}^{n_1}\cdots\sum_{k_K=0}^{n_K}\frac1{\sqrt{\pi}} \frac{\mathcal{A}}{2^{\sum_{k=1}^K \frac{n_k-k_k}{2}+\frac12}} \\ & \times  \int_0^\infty \frac{\Gamma \left( \alpha_D,\frac{W_D}{ \gamma x}\right)}{x^{\sum_{k=1}^K \frac{n_k-k_k}{2}+1}}   \exp\left( -\frac{\left( \sum_{k=1}^K {\sqrt{W_k}}\right)^2}{2x}\right) \\ & \times {\rm D}_{1+\sum_{k=1}^K n_k-k_k} \left( \frac{\sum_{k=1}^K \sqrt{2W_k}}{\sqrt{x}}\right) dx,
    \end{split}
\end{equation}
{where $\alpha_D$ represents the shaping parameter of the shadowing distribution for the desired link.}  {It is worth noting that, although \eqref{eq:USE_general} involves an improper integral over $(0,\infty)$, the corresponding integrand is smooth and decays rapidly, which ensures convergence and facilitates accurate numerical evaluation through conventional numerical integration techniques.} 

\subsubsection{Asymptotic Results}  {In what follows, a simpler closed-form expression for the CP will be obtained, which can also be used to analyze the OP. More specifically,} assuming lower values of the transmit power for the serving aerial BS, i.e., $P_m<20~\mathrm{dBm}$, the upper incomplete gamma function included in \eqref{eq:USE_general} simplifies as follows  {$\Gamma(a,z)\simeq \Gamma(a)-\frac{z^a}{a}$} \cite[eq. (06.06.06.0004.01)]{E:wolfram}. Using this approximation, making a change of variables of the form $y=\frac{A}{\sqrt{x}}$, with $A\in \mathbb{R}$, results in the following type of integral
\begin{equation}\label{eq:integral}
\begin{split}
    \mathcal{I}&=\int_0^\infty t^{a-1}\exp\left( -bt^2\right) {\rm D}_v(t)dt \,{=}\, \frac{2^{\frac{v-a}2} \sqrt{\pi} \Gamma(a)}{\Gamma\left( \frac{a-v+1}{2}\right)} \\ & \times {\rm \null_2F_1}\left(\frac{a}{2},\frac{a+1}{2};\frac{a-v+1}{2};\frac12-2b \right).
\end{split}
\end{equation}
In \eqref{eq:integral}, {the second equality} is due to \cite[eq. (07.41.21.0014.01)]{E:wolfram}, while $\null_2F_1(\cdot)$ denotes the Gauss hypergeometric function \cite[eq. (9.100)]{B:Ryzhik_book}. Moreover, based on the above solution, using {\cite[eq. (07.23.03.0001.01)]{E:wolfram}}, and after some mathematical simplifications yields the following closed-form expression  {for the CP}
    \begin{equation}\label{eq:high_SNR1}
    \begin{split}
   &     C_p\simeq
        \frac{\left( \prod_{k=1}^K \frac{\sqrt{\pi}}{\Gamma(\alpha_k)}\right)}{{\Gamma(\alpha_D)^2 }} \sum_{k_1=0}^{n_1}\cdots\sum_{k_K=0}^{n_K}  \frac{\mathcal{A}}{ \alpha_D}\\ & \times \frac{2\Gamma\left[2 \left(\alpha_D +\sum\limits_{k=1}^K \frac{n_k-k_k}{2}\right)\right]}{ (\gamma/W_D)^{\alpha_D}\left(2\sum\limits_{k=1}^K {\sqrt{W_k}}\right)^{2\left(\alpha_D +\sum\limits_{k=1}^K \frac{n_k-k_k}{2}\right)}}  .
    \end{split}
\end{equation}
\begin{figure}
  \centering
    \includegraphics[width=2.8in]{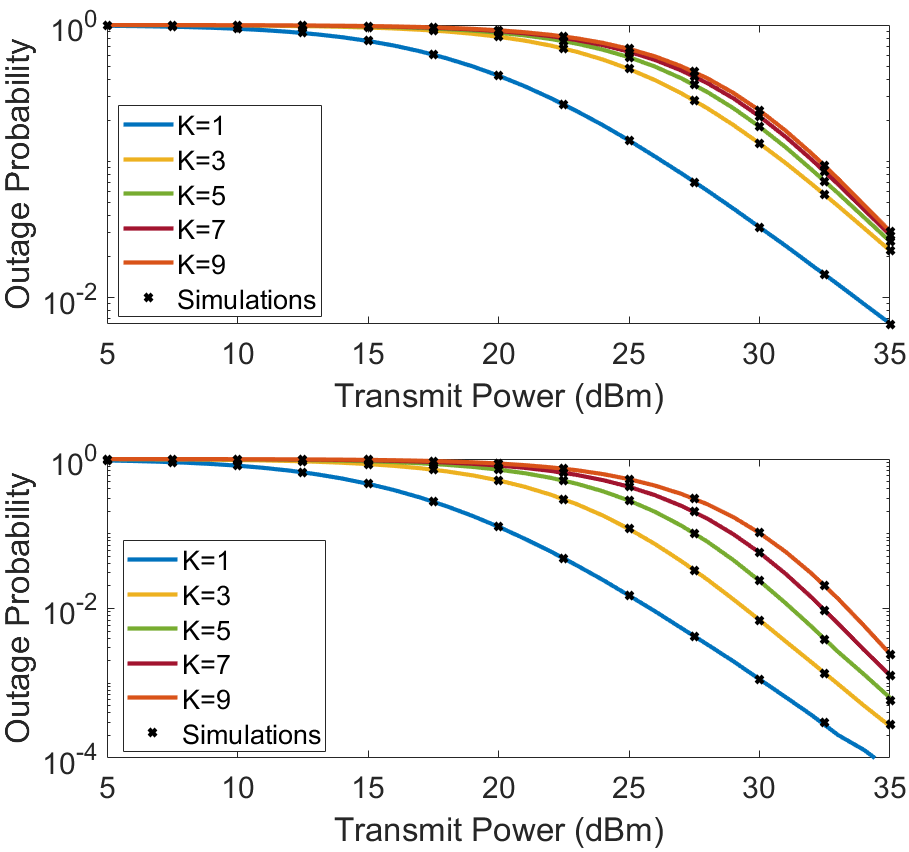}
    \caption{OP vs the transmit power for different numbers of interfering UAVs: In i.n.d. shadowing, OP increases.}
    \label{fig:fig1}
\end{figure}
As it is shown in \eqref{eq:high_SNR1}, the channel reliability slope depends only on the shadowing severity coefficient. 

\subsubsection{Numerical Results}
{Analytical results are presented that demonstrate how the system performance varies with key system and channel parameters.  {The simulation methodology is based on direct Monte Carlo (MC) generation of gamma RVs with the desired shaping and scaling parameters, followed by inversion to obtain IG samples, summation over the required number of links, computation of the instantaneous SNR/SIR, and empirical estimation of the corresponding OP/CP through sample averaging over a large number of realizations (e.g., $10^6$)}. Moreover, all numerical integrations were evaluated in \textit{Mathematica} using \texttt{NIntegrate} with its default error-control settings.
}  {Unless otherwise stated, the following assumptions are adopted: \(n=2\), \(P_k=20~\mathrm{dBm}\), and \(G_m=G_k=8~\mathrm{dBi}\).} In Fig.~\ref{fig:fig1}, the OP is plotted as a function of the transmit power $P_m$ for different numbers of interfering UAVs. For evaluating these results, light shadowing conditions are assumed, i.e., { {$\alpha_D=1.5$, while  $d_m=80~\mathrm{m}$}. Moreover, in Fig.~\ref{fig:fig1}.a, i.n.d. shadowing coefficients are assumed, with $d_k \in D$, where {$D=\{90, 95, 100, 105, 110, 115, 120, 125, 130\}~\mathrm{m}$}  {and $a_k \in A$, where {$A=\{1.5, 1.5, 1.5, 2.5, 2.5, 2.5, 3.5, 3.5, 3.5\}$}}. It is shown that the performance degrades as the number of interferers increases, as expected, albeit with a reduced rate of decline. In Fig.~\ref{fig:fig1}.b, for comparison purposes, another scenario is investigated, in which i.i.d. shadowing coefficients are assumed with $d_k=\mu_D = \frac{1}{|D|}\sum\limits_{d_k \in D}d_k$  {and $a_k=\mu_A = \frac{1}{|A|}\sum\limits_{a_k \in A}a_k$}. In that figure, it is also shown that the performance decreases as $K$ increases. However, it is also shown that the curves are more widely separated and extend to lower outage levels. In Fig.~\ref{fig:fig4},  {assuming $a_1=1.5,a_2=2.5$ and $d_k=d_m=107.5~\mathrm{m}$} the CP is plotted as a function of $P_m$ for different values of {$\alpha_D$}. It is clearly shown  {that for the same transmit power, different values of $\alpha_D$ result in significantly different CP values\footnote{ {It is worth noting that, under the IG formulation adopted here, the mean of the shadowing RV depends on $\alpha_k$, i.e., $\mathbb{E}[I_k]=W_k/(\alpha_k-1)$ for $\alpha_k>1$. Hence, for fixed $W_k$, increasing $\alpha_k$ reduces the average received power.}} Therefore, in Fig.~\ref{fig:fig4}, increasing $\alpha_D$ reduces the mean value of the desired-link shadowing RV and, consequently, the average useful received power, which explains the observed performance degradation. Moreover, the very small relative errors reported in the table included in Fig.~\ref{fig:fig4} confirm that (11) provides an accurate and practically useful approximation of the exact system behavior across the considered operating points, especially for lower transmit power values, i.e., below $10~\mathrm{dBm}$.}
Finally, an excellent agreement is observed between the analytical and simulated results, especially for higher transmit power values, whereas small deviations are observed only in the very low-outage regime.

\begin{figure}
  \centering
    \includegraphics[width=2.6in]{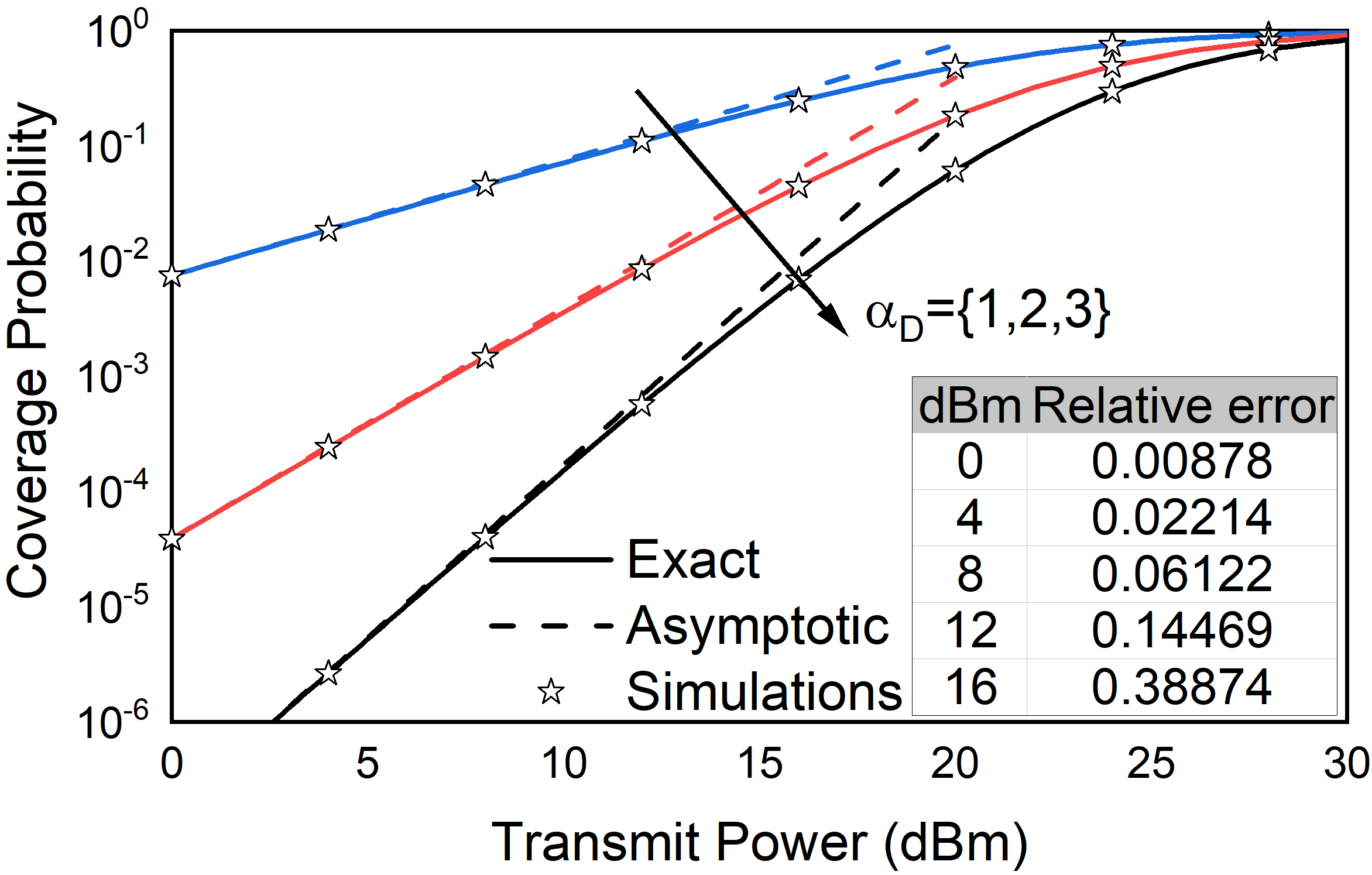}
    \caption{CP vs the transmit power: Shadowing severity has an important impact to the performance.}
    \label{fig:fig4}
\end{figure}

\subsection{Coverage Analysis for HSR Communication Networks}
Here, a coverage analytical framework is presented for an HSR communication scenario that leverages CA. BSs are uniformly spaced beside the track, each mounting two outward‑facing directional antennas that create narrow strip‑shaped cells. A train access point (TAP) installed above the lead car aggregates onboard traffic, turning every user–BS transmission into a two‑hop path: passenger to TAP (front‑end link) and TAP to BS (backhaul link) \cite{lin2021coverage}. Because service quality is dominated by backhaul, analysis is concentrated there \cite{lin2020coverage}. Traffic intensity varies geographically; consequently, sections with high demand, e.g., near stations—exploit CA with two non‑contiguous component carriers $f_1$ and $f_2$, whereas ordinary segments operate on single carrier $f_1$. Different Hata‑based path‑loss models characterize each carrier’s propagation under various terrain conditions.

The received instantaneous power at the TAP from carrier $k$ (with $k \in \{1,2\}$) is given by
\begin{equation}
    P_{r_k}=\underbrace{\frac{K_sP_T}{10^{A_k/10}r^{B_k/10}}}_{S_k} I_k,
\end{equation}
where 
\begin{displaymath}
\begin{split}
A_k=&\Delta_1+74.52+26.16 \log_{10}(f_k)-13.82\log_{10}(h_b)\\ & -3.2\left(\log_{10}(11.75h_m)\right)^2    
\end{split}
\end{displaymath}
and
\begin{displaymath}
B_k=\left[44.9-6.55\log_{10}(h_b)+\Delta_2\right].  
\end{displaymath}
In the above expressions,  $h_b, h_m$ denote the effective antenna height of the BS and the TAP, respectively, $r$ denotes the transmit distance between the BS and TAP antennas,  {and $\Delta_1,\Delta_2$ are the correction factors for the model \cite{lin2021coverage}}. Moreover, $P_T$ denotes the transmit power and $K_s$ is a
unitless constant that depends on the antenna characteristics.

Based on the above presented path loss model and assuming that CA is employed, the total received SNR of the TAP at position $r$ is given by
\begin{equation}\label{eq:SNR_coverage_def}
    \gamma_{\rm out}=\sum_{k=1}^2 \frac{\eta_k P_{r_k}}{N_k}=\sum_{k=1}^2 \frac{\eta_kS_kI_k}{N_k}=\sum_{k=1}^2 \eta_kW_kI_k,
\end{equation}
where $N_k$ denotes the noise variance at the TAP over carrier $k$, while $\eta_1=1, \eta_2\in\{0,1\}$. For 
$\eta_2=1$, two backhaul links are assumed, with CA yielding a total SNR equal to their sum, while for $\eta_2=0$, only the $f_1$ link is used, and the total SNR is that of the single link. Based on \eqref{eq:SNR_coverage_def}, the CDF of the output SNR of the communication system is given by $F_{\gamma_{\rm out}} (\gamma)=F_{\bf I}(\gamma)$, with $W_k=\frac{K_sP_T}{10^{A_k/10}r^{B_k/10}N_k}$. 

\subsubsection{Asymptotic Results}
Assuming higher values of the transmit power, a simplified expression can be obtained for the CDF of the SNR of the scheme under consideration.  {More specifically, starting from \eqref{eq:CDF_sum}, assuming higher values of $W_k$} and using \cite[eq. (07.41.06.0013.01)]{E:wolfram}, the parabolic cylinder function can asymptotically be approximated by $D_v(x) \simeq  x^v \exp\left( -\frac{x^2}{4}\right)$.  {The adopted approximation corresponds to the leading-order truncation of the asymptotic expansion of the parabolic cylinder function, which is sufficient to capture the dominant behavior in the large-argument regime while preserving analytical tractability.} Based on this approximation and after some mathematical simplifications, the CDF of $\gamma_{\rm out}$  {simplifies to the following closed-form expression}
\begin{equation}\label{eq:CDF_sum_asy}
\begin{split}
   & F_{\gamma_{\rm out}} (\gamma)\simeq \left( \prod_{k=1}^K \frac{\sqrt{\pi}}{\Gamma(\alpha_k)}\right) \sum_{k_1=0}^{n_1}\cdots\sum_{k_K=0}^{n_K} \frac{\mathcal{A}}{\sqrt{\pi}} \\ & \times    \frac{\left( \sum_{k=1}^K {\sqrt{W_k}}\right)^{\sum_{k=1}^K n_k-k_k-1}}{\gamma^{\sum_{k=1}^K n_k-k_k-1/2}}  \exp\left( -\frac{\left( \sum_{k=1}^K {\sqrt{W_k}}\right)^2}{\gamma}\right)  .
\end{split}
\end{equation}

\subsubsection{Numerical Results}
{The performance of the scheme under consideration is evaluated in terms of the CDF of ($\gamma_{\mathrm{out}}$) and the edge CP, while MC simulations with more than ($10^6$) samples are also conducted for comparison purposes.}
CP is evaluated at the cell edge ($r=D$). In Fig.~\ref{fig:fig2},  {the OP, i.e., $F_{\gamma_{\mathrm{out}}}(\gamma_{\mathrm{th}})=\Pr\{\gamma_{\mathrm{out}}\leq \gamma_{\mathrm{th}}\}$, is plotted versus the transmit power $P_T$ for different HSR propagation scenarios, with the main simulation parameters set as in \cite{lin2021coverage}: $((f_1,f_2)=(350~\mathrm{MHz},400~\mathrm{MHz})), (h_b=40~\mathrm{m}), (h_m=4~\mathrm{m}), \textrm{ and } (N_k=10^{-5})$.
 Moreover, it is also assumed that CA is employed in all scenarios investigated, while {$r=1~\mathrm{km}$, $\gamma_{\rm th}=0~\mathrm{dB}$},  {and $a_1=1.5,$ $a_2=2.5$}. It is shown that OP decreases monotonically as the transmit power increases. Among the HSR propagation scenarios, the viaduct yields the smallest OP, whereas the cutting exhibits the highest. Overall, HSR performance is strongly scenario-dependent and larger cells amplify these differences. Moreover, in the same figure, the asymptotic curves derived from \eqref{eq:CDF_sum_asy} are also plotted, demonstrating an excellent match with the exact results, particularly at higher transmit power levels.

Moreover, in Fig.~\ref{fig:fig3}, the edge CP is plotted as a function of radius $D$ for different power levels. Moreover, the performance of the CA-based systems is compared with those without CA. As clearly shown, across all scenarios, the system with CA consistently outperforms its counterpart without CA. Moreover, at higher values of the transmit power, for small cell radii, the cell edge can be almost fully covered, resulting in comparable performance between systems with and without CA. As the radius increases to moderate values, the edge CP gap between the two systems widens with larger $ 
D$. For instance, at 
{$P_T=40~\mathrm{dBm}$} and edge CP equal to $3\cdot10^{-1}$, CA extends the cell radius by almost {$0.5~\mathrm{km}$}. At large cell radii, however, the edge CP for both systems approaches zero, indicating that the cell edge is hardly covered. Finally, it is noted that in all figures, the simulation results are in excellent agreement with the numerically evaluated ones.

\begin{figure}
  \centering
    \includegraphics[width=2.6in]{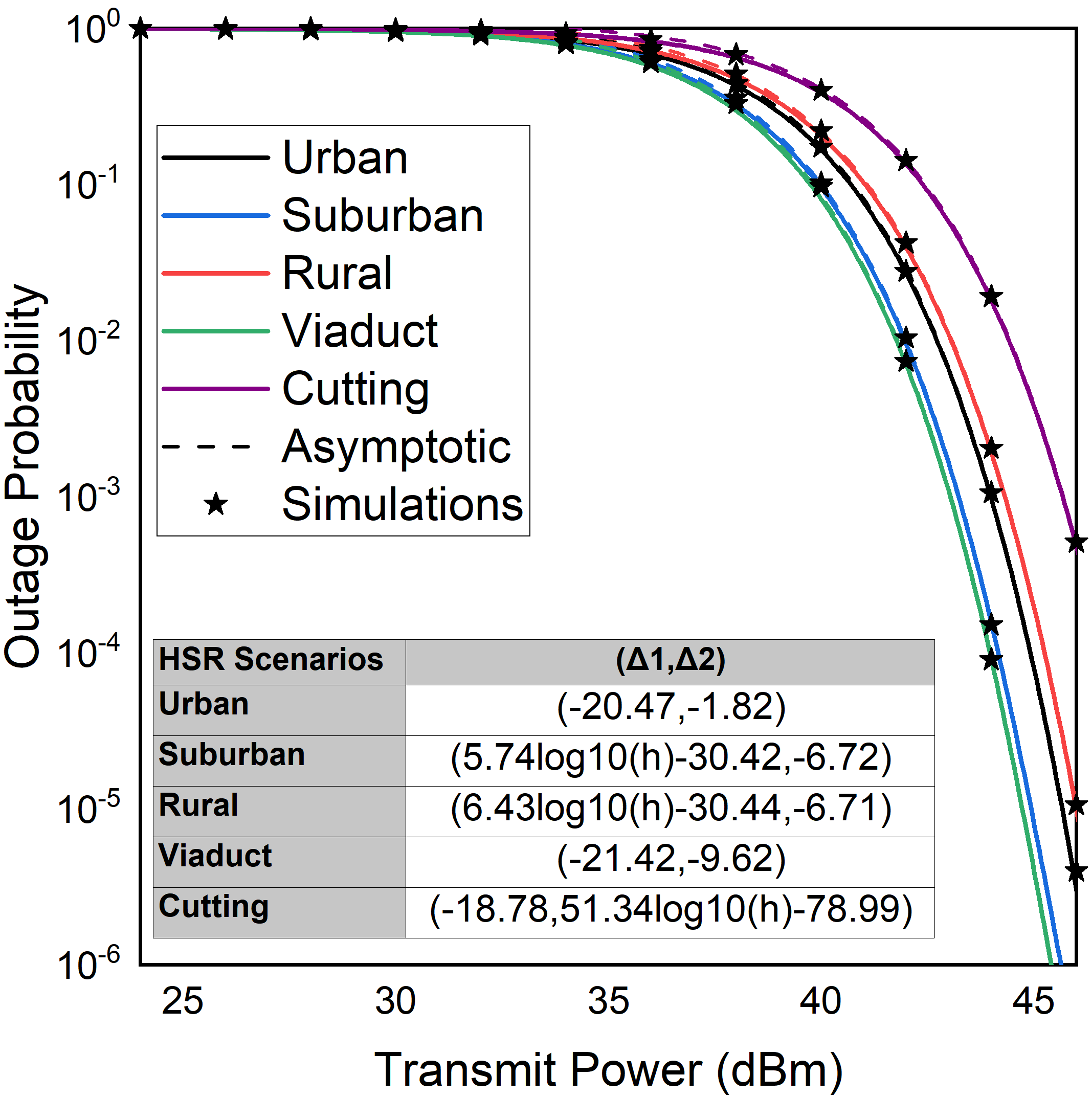}
    \caption{ {OP vs  $P_T$ for different HSR propagation scenarios, with $r=1$ km and $\gamma_{\mathrm{th}}=0$ dB (in this table $h \equiv h_b $).}}
    \label{fig:fig2}
\end{figure}

\begin{figure}
  \centering
    \includegraphics[width=2.5in]{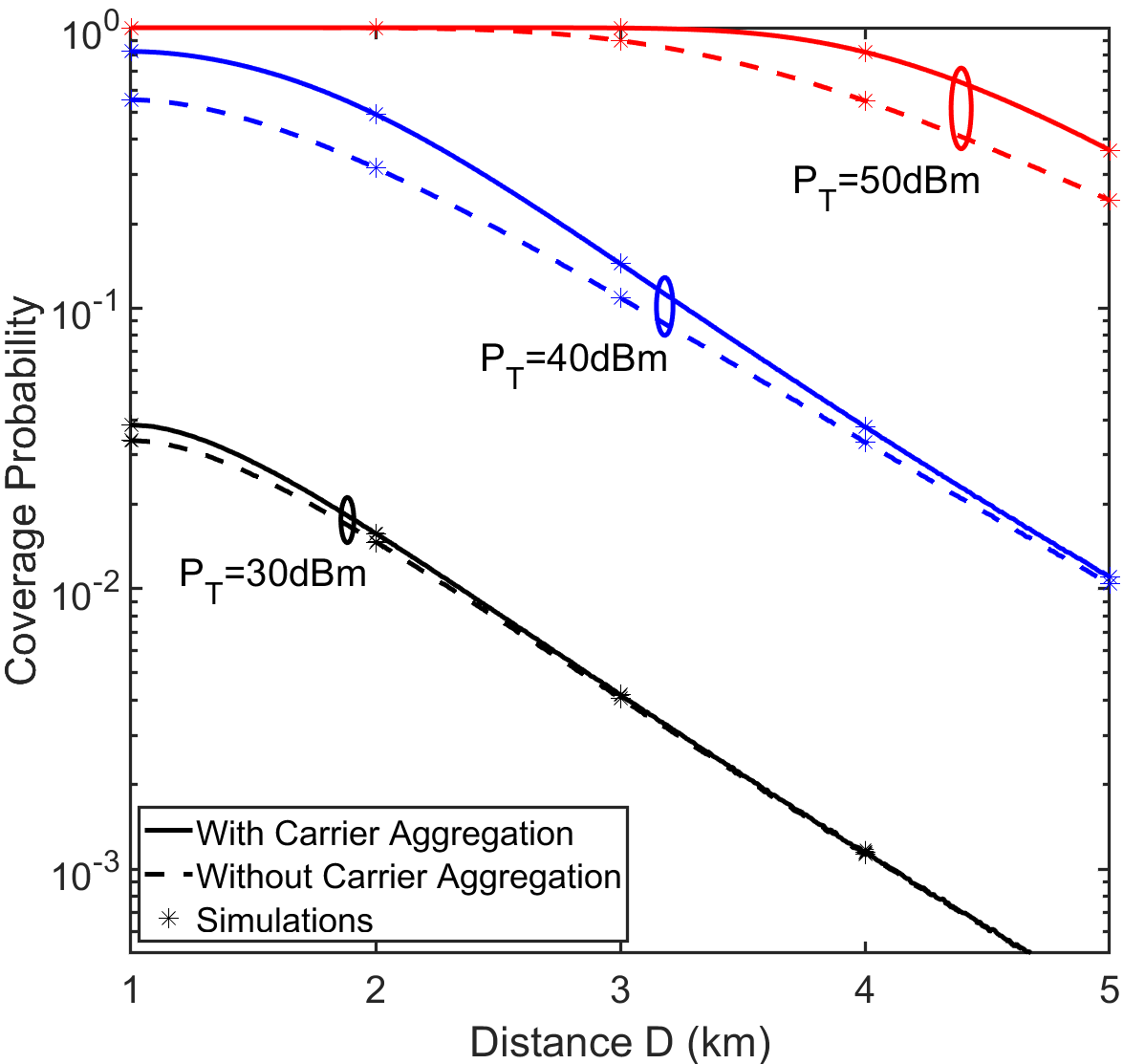}
    \caption{CP vs distance for different transmit powers: CA improves performance in moderate values of the distance.}
    \label{fig:fig3}
\end{figure}

\section{Conclusion}
{ {This work presents an exact and numerically robust analytical framework for the sum of i.n.d. IG RVs, yielding closed-form PDF/CDF expressions that avoid multi-fold integration and do not rely on slowly convergent series. The derived formulas are computationally efficient and enable reliable outage/coverage evaluation over broad parameter ranges, while the proposed asymptotic approximations provide additional insights. The practical value of the analysis is demonstrated in two representative scenarios. In the UAV-assisted network case, the framework produces tractable outage expressions and reveals the impact of the number of interfering UAVs, as well as the performance difference between i.n.d. and i.i.d. shadowing assumptions. In the HSR CA case, the results quantify the dependence of coverage on propagation conditions and confirm the coverage gains offered by aggregation, particularly for moderate cell radii. As a future research direction, it would be interesting to investigate the corresponding statistics of the sum of correlated IG RVs. Such analysis could be applied to UAV-assisted communications, where shadowing is expected to exhibit spatial correlation both across distinct UAV and along successive waypoints of a single UAV trajectory \cite{bithas2024uav}.
}

\appendix 
\section{Proof of Theorem}
In this appendix, the proofs of Theorem~1 and Lemma~1 are given. Starting from the Laplace transform of $I_k$ in \eqref{eq:MGF_Ik}, assuming $\alpha_k=n_k+1/2$, and using \cite[eq. (8.468)]{B:Ryzhik_book}, yields
\begin{equation}\label{eq:MGF_Ik_2nd}
\begin{split}
    M_{I_k} (s)&= \frac{\sqrt{\pi}}{\Gamma(\alpha_k)} \exp\left( -{2\sqrt{sW_k}}\right) \\ & \times \sum_{k_k=0}^{n_k} \frac{(n_k+k_k)! (sW_k)^{n_k-k_k}}{k_k! (n_k-k_k)! 2^{2k_k} }.
    \end{split}
\end{equation}
Since the $I_k$'s are independent RVs, the Laplace transform of $\mathbf{I}$ is given by \cite{moschopoulos1985distribution}
\begin{equation}\label{eq:MGF_def}
    M_{\bf{I}}(s)=\prod_{k=1}^K M_{I_k} (s).
\end{equation}
Substituting \eqref{eq:MGF_Ik_2nd} in \eqref{eq:MGF_def} and after some  {mathematical} simplifications yields
\begin{equation}\label{eq:MGF}
\begin{split}
    M_{\bf{I}}(s)&=\left( \prod_{k=1}^K \frac{\sqrt{\pi}}{\Gamma(\alpha_k)}\right) \sum_{k_1=0}^{n_1}\cdots\sum_{k_K=0}^{n_K} \mathcal{A}\\ & \times s^{\sum_{k=1}^K \frac{n_k-k_k}{2}}   \exp\left( -2\sum_{k=1}^K {\sqrt{sW_k}}\right).
\end{split}
\end{equation}
Each term in \eqref{eq:MGF} has the form $M_{\mathbf{I}}(s)=A s^B\exp(-C\sqrt{s})$, where $ B=\frac{1}{2}\sum_{k=1}^{K}(n_k-k_k), \qquad C=2\sum_{k=1}^{K}\sqrt{W_k},$ and coefficient $A$ is related to the multiple sum. The PDF and CDF follow from the inverse Laplace transforms
\begin{equation}\label{eq:invLaplace_pdf_cdf}
f_{\bf I}(x)=\mathcal{L}^{-1}\!\left\{ M_{\bf I}(s)\right\}, \qquad
F_{\bf I}(x)=\mathcal{L}^{-1}\!\left\{ \frac{M_{\bf I}(s)}{s}\right\}.
\end{equation}
Here, $\mathcal{L}^{-1}\{\cdot\}$ denotes the inverse Laplace
transform. Using \cite[eq. (2.2.1/10)]{B:prudnikov_book5}, the following exact generic expressions for $f_{\bf{I}}(x)$ and $F_{\bf{I}}(x)$ are obtained
 {\begin{subequations}\label{eq:generic_pdf_cdf}
\begin{align}
f_{\bf I}(x) &=
\frac{A/\sqrt{\pi}}{2^{B+\frac{1}{2}}\,x^{B+1}}
\exp\!\left(-\frac{C^2}{8x}\right)
{\rm D}_{2B+1}\!\left(\frac{C}{\sqrt{2x}}\right),
\label{eq:generic_pdf}\\[2mm]
F_{\bf I}(x) &=
\frac{A}{2^{B-\frac{1}{2}}\,x^{B}\sqrt{\pi}}
\exp\!\left(-\frac{C^2}{8x}\right)
{\rm D}_{2B-1}\!\left(\frac{C}{\sqrt{2x}}\right).
\label{eq:generic_cdf}
\end{align}
\end{subequations}}
Finally, substituting the corresponding values of $A$, $B$,
and $C$ into \eqref{eq:generic_pdf} and
\eqref{eq:generic_cdf} yields \eqref{eq:PDF_sum} and
\eqref{eq:CDF_sum}, thereby completing the proof.


%



\ifCLASSOPTIONcaptionsoff
  \newpage
\fi



%
\balance
\bibliographystyle{IEEEtran}
\bibliography{IEEEabrv,references}

%




\end{document}